\documentclass[12pt, 
]{article}
\usepackage{booktabs}

\usepackage[utf8]{inputenc}
\usepackage{amssymb}
\usepackage{amsmath}

\usepackage{geometry}
\newtheorem{theorem}{Theorem}

\newtheorem{fact}{Fact}
\newtheorem{corollary}{Corollary}
\newtheorem{definition}{Definition}
\newtheorem{example}{Example}
\newtheorem{lemma}{Lemma}
\newtheorem{remark}{Remark}

\newenvironment{proof}[1][Proof]{\emph{#1.} }{\  \hfill $\square $ \vspace{5 pt}}

\usepackage{natbib}
\usepackage{amssymb}
\usepackage[dvips]{graphics}
\usepackage{amsmath}

\usepackage{mathpazo}
\usepackage{enumerate}

\usepackage{amsfonts}

\usepackage{amssymb}

\usepackage{enumerate}

\usepackage{mathrsfs}

\usepackage[usenames]{color}
\usepackage{cancel}

\usepackage{pgf,tikz}

\usetikzlibrary{decorations.pathreplacing,angles,quotes}
\usetikzlibrary{arrows.meta}
\usetikzlibrary{arrows, decorations.markings}
\tikzset{myptr/.style={decoration={markings,mark=at position 1 with %
       {\arrow[scale=2,>=stealth]{>}}},postaction={decorate}}}

\usepackage{hyperref}
\hypersetup{
    colorlinks,
    citecolor=blue,
    filecolor=black,
    linkcolor=blue,
    urlcolor=blue
}

\usepackage{pgf,tikz}
\usetikzlibrary{arrows}

\usepackage{fancyhdr}

\usepackage{soul}

\usepackage{emptypage}

\newcommand*\samethanks[1][\value{footnote}]{\footnotemark[#1]}

\usepackage[T1]{fontenc}
\DeclareFontFamily{T1}{calligra}{}
\DeclareFontShape{T1}{calligra}{m}{n}{<->s*[1.44]callig15}{}
\DeclareMathAlphabet\mathcalligra   {T1}{calligra} {m} {n}

\usepackage{multirow}
\usepackage{array}
\usepackage{mathtools}
\usepackage{arydshln}
\usepackage{threeparttable}
\usepackage{booktabs,caption}

\usepackage{ifthen}
\newboolean{showcomments}
\setboolean{showcomments}{true}

\newcommand{\pablo}[1]{  \ifthenelse{\boolean{showcomments}}
{\textcolor{green!50!black}{(T: #1)}}{}}
\newcommand{\marcelo}[1]{\ifthenelse{\boolean{showcomments}}
{\textcolor{red}{(M: #1)}}{}}
\newcommand{\agustin}[1]{  \ifthenelse{\boolean{showcomments}}
{\textcolor{blue!50!black}{(T: #1)}}{}}

\begin{document}

\title{Obvious manipulations of tops-only voting rules%
\thanks{%
We thank Jordi Massó for his very detailed comments. We acknowledge financial support
from UNSL through grants 032016, 030120, and 030320, from Consejo Nacional
de Investigaciones Cient\'{\i}ficas y T\'{e}cnicas (CONICET) through grant
PIP 112-200801-00655, and from Agencia Nacional de Promoción Cient\'ifica y Tecnológica through grant PICT 2017-2355.}}


\author{R. Pablo Arribillaga\thanks{
Instituto de Matem\'{a}tica Aplicada San Luis, Universidad Nacional de San
Luis and CONICET, San Luis, Argentina, and RedNIE. Emails: \href{mailto:rarribi@unsl.edu.ar}{rarribi@unsl.edu.ar} 
and \href{mailto:abonifacio@unsl.edu.ar}{abonifacio@unsl.edu.ar} 
} \and Agustín G. Bonifacio\samethanks[2] 
}

\date{\today}

\maketitle

\begin{abstract}
In a voting problem with a finite set of alternatives to choose from, we study the manipulation of tops-only rules. Since 
all non-dictatorial (onto) voting rules are manipulable when there are more than two alternatives and all preferences are allowed, we look for rules in which manipulations are not obvious. First, we show that a rule does not have obvious manipulations if and only if when an agent vetoes an alternative  it can do so with any preference that does not have such alternative in the top.
Second, we focus on two classes of tops-only rules: (i) (generalized) median voter schemes, and (ii) voting by committees. For each class, we identify which rules do not have obvious manipulations on the universal domain of preferences. 


\bigskip

 \bigskip


\end{abstract}

\section{Introduction}

Voting rules are systematic procedures that allow a group of agents to select an alternative, among
many, according to their preferences. Within desirable properties a voting rule may satisfy, the concept of strategy-proofness has played a central
role for studying the 
strategic behavior of the agents. A voting rule is strategy-proof if it is always in the best interest of the agents
to reveal their true preferences. Unfortunately, Gibbard-Satterthwaite’s celebrated theorem states that, outside of a dictatorship, there is no strategy-proof (onto) voting rule when
more than two alternatives and all possible preferences over alternatives are considered \citep{gibbard1973manipulation, satterthwaite1975strategy}. To circumvent  this impossibility
result, two main approaches have been taken. The first approach restricts the domain of preferences that agents can have
over alternatives \citep[see][and references therein]{barbera2011strategyproof}. The second approach considers weakenings of strategy-proofness, and it has been an active field of research in recent years. The idea  underlying this approach is that even though  manipulations are pervasive,  agents may not realize they can manipulate a rule  because they lack information about others' behavior or they are cognitively limited. 

\cite{troyan2020obvious} introduce the concept of obvious manipulation in the context
of market design. They assume that an agent knows the possible outcomes of the mechanism
conditional on his own declaration of preferences, and define a deviation from the truth
to be an obvious manipulation if  the best possible outcome under the deviation
is strictly better than the best possible outcome under truth-telling, or the worst possible
outcome under the deviation is strictly better than the worst possible outcome under
truth-telling. A mechanism that does not allow any obvious manipulation is called not obviously 
manipulable.

In this paper we study (not) obvious manipulation of voting rules when a finite set of alternatives is involved. We focus on tops-only rules: rules that only consider agents' top alternatives in order to select a social choice. Obvious manipulations of non-tops-only rules have been thoroughly studied elsewhere \citep[see][]{aziz2021obvious}. The importance of studying  obvious manipulations of  tops-only rules is three-fold. Firstly, because of their simplicity,  tops-only voting rules are important rules on their own right and are both useful in practice and extensively studied in the literature. Secondly, assume we have a class of strategy-proof rules defined on a restricted domain of preferences and we want to study whether those rules are obviously manipulable in the universal domain of preferences. Then, it is natural to require tops-onlyness since, under mild assumptions, strategy-proofness implies tops-onlyness \citep[see, for example,][]{barbera1991voting,chatterji2011tops, weymark2008strategy}. Thirdly, obvious manipulations  allow us to discriminate among  different tops-only rules, while other recently studied weakenings of strategy-proofness are incapable to do this because, under tops-onlyness, those concepts become equivalent to strategy-proofness \citep[see, for example,][]{arribillaga2022regret}.

Our main result gives a characterization of not obviously manipulable rules in terms of  veto power of the agents. An agent vetoes an alternative if there is a preference report of the agent that forces the rule to never select such alternative. The veto is strong if the report of \emph{any} preference with top different from the alternative forces the rule to never select it.  Theorem \ref{Main Theorem} states that, within tops-only rules, not obvious manipulation is equivalent to each veto being a strong veto.

Next, we apply our main result to study well-known classes of tops-only voting rules when they are considered on the universal domain of preferences: (generalized) median voter schemes and voting by committees. 

First, consider a problem  where the set of alternatives $X$ is an ordered set. Without loss of generality, assume $X=\{a,a+1,a+2,\ldots, b\}\subseteq \mathbb{N}$ and $%
b=a+(m-1).$ For this problem, \cite{moulin1980strategy} characterizes all strategy-proof and tops-only (onto) rules on the  domain of single-peaked preferences\footnote{An agent’s preference is single-peaked if there is a top alternative that is strictly preferred to all other alternatives and at each of the two sides of the top alternative the preference is monotonic, increasing in the left, and decreasing in the right.} as the class of all generalized median voter schemes. \cite{moulin1980strategy} also characterizes the subclass of
median voter schemes as the set of all strategy-proof, tops-only, and anonymous (onto) rules on the domain of single-peaked preferences. A median voter scheme can be identified
with a vector $\alpha = (\alpha_1, \ldots, \alpha_{n-1})$ of $n - 1$ numbers in $X$, where $n$ is the cardinality of
the set of agents $N$ and $\alpha_1\leq \alpha_2\leq \ldots\leq \alpha_{n-1}$. Then, for each preference profile, the median
voter scheme identified with $\alpha$ selects the alternative that is the median among the $n$
top alternatives of the agents and the $n-1$ fixed numbers $\alpha_1, \ldots, \alpha_{n-1}$. 
Since $2n - 1$ is an odd number, this median always exists and belongs to $X$. Generalized
median voter schemes constitute non-anonymous extensions of median voter schemes. 

When the designer cannot guarantee  that the domain restriction (single-peakedness) is met and the full domain of preferences has to be considered, then strategy-proofness no longer holds. For this reason,  it is important to identify which rules within these families obey   the less demanding property of non-obvious manipulability. In Theorem \ref{teorema median}, we show that a median voter scheme is not obviously manipulable if and only if  $\alpha_1 \in \{a, a+1\}$ 
 and $\alpha_{n-1}\in \{b-1,b\}.$ A similar condition applied to the extremal fixed ballots (for each agent) in the monotonic family of fixed ballots associated with  generalized median voter schemes characterizes those that are   not obviously manipulable (Theorem \ref{teo general median}).

Now, consider a problem in which  agents have to choose a\textit{\ subset of objects} from a set $K$ (with $\left\vert K\right\vert \geq 2$)$.$ Then, in this case, $X=2^{K}$ and elements of $X$ are subsets of $K$. A generic element of $K$ is
denoted by $k$.  As an example, think of objects as candidates
to become new members of a society that have to be elected by the current members of the society. \cite{barbera1991voting} characterize, on the restricted domain of separable preferences,\footnote{An agent’s preference is separable on the family of all subsets of objects if the division between good objects and bad objects guides the ordering of all subsets in the sense that adding a good object
to any set leads to a better set, while adding a bad object leads to a worse set.} the family of all strategy-proof (onto) rules as the class of voting by
committees. Following \cite{barbera1991voting}, a voting
by committees is defined by specifying for each object $k \in K$ a monotonic family
of winning coalitions $\mathcal{W}_k$ (a committee). Then, the choice of the subset of objects
made by a voting by committees at a preference profile is done object-by-object
as follows. Fix a voting by committees $\mathcal{W}=\{\mathcal{W}_k\}_{k \in K}$ and a preference profile,
and consider object $k$. Then, $k$ belongs to the chosen set (the one selected by $\mathcal{W}$ at
the preference profile) if and only if the set of agents whose best subset of objects
contains $k$ belongs to the committee $\mathcal{W}_k$.\footnote{Voting by committees can be seen
as a family of extended majority voting (one for each object $k$), where the two
alternatives at stake are whether or not $k$ belongs to the collectively chosen subset of objects.} Observe that voting by committees are in fact tops-only, so tops-onlyness is implied by strategy-proofness on the domain of separable preferences. 

Again, if the domain restriction (separability, in this case) is not guaranteed to be met, identifying which rules are not obviously manipulable can be helpful. In Theorem \ref{teo committees}, we show that a non-dictatorial voting by committees is not obviously manipulable if and only if no agent is a vetoer. In terms of the committees defining the rule, this is equivalent to say that, for each object: (i) no agent belongs to all the coalitions in the committee, and (ii) no singleton coalition belongs to the committee. When anonymity is added to the picture,  voting by committees  simplifies to voting by quota.  In this case, for each $k \in K$, there is a number $q_k$ such that $\mathcal{W}_k$ is the set of all coalitions with cardinality at least $q_k$. We prove that non-obvious manipulability is equivalent to  each committee having quota between 2 and $n-1$ (Corollary \ref{cor quota}).

The paper of \cite{aziz2021obvious} is the closest to ours and, to the best of our knowledge, is the first one that applies \cite{troyan2020obvious} notion of obvious manipulation to the context of voting. \cite{aziz2021obvious} present a general sufficient condition for
a voting rule to be not obviously manipulable. However, they focus on non-tops-only rules. They show that Condorcet consistent as
well as some other strict scoring rules are not obviously manipulable. Furthermore, for
the class of $k$-approval voting rules, they give necessary and sufficient conditions for obvious
manipulability. Other recent papers that study obvious manipulations, in contexts other than voting, are 
 \cite{ortega2022obvious} and  \cite{psomas2022fair}.
















The rest of the paper is organized as follows. The model and the concept of obvious manipulations are introduced in Section \ref{section prelim}. In Section \ref{section main}, we present the main result of our paper that, under tops-onlyness,  characterizes  non-obvious manipulable voting rules. Section \ref{section applic} deals with applications: in Subsection \ref{subsection median} we study (generalized) median voter schemes, and in Subsection \ref{subsection committees} we study voting by committees. To conclude, some final remarks are gathered in Section \ref{section final}. 

\section{Preliminaries}\label{section prelim}
\subsection{Model}
A set of \textit{agents }$N=\{1,\ldots ,n\}$, with $n\geq 2$, has to choose
an alternative from a finite and given set $X$ (with cardinality $\left\vert
X\right\vert =m\geq 2$ )$.$ Each agent $i\in N$ has a strict \textit{%
preference} $P_{i}$ over $X.$ Denote by $t(P_{i})$ to the
best alternative according to $P_{i}$, called the 
\textit{top }of $P_{i}$. We denote by $R_{i}$ the weak preference over $X$
associated to $P_{i};$ \textit{i.e., }for all $x,y\in X$, $xR_{i}y$ if and
only if either $x=y$ or $xP_{i}y.$ Let $\mathcal{P}$ be the set of all
strict preferences over $X.$ A (preference)\textit{\ profile} is a $n$-tuple 
$P=(P_{1},\ldots ,P_{n})\in \mathcal{P}^{n},$ an ordered list of $n$
preferences, one for each agent. Given a profile $P$ and an agent $i,$ $%
P_{-i}$ denotes the subprofile in $\mathcal{P}^{n-1}$ obtained by deleting $%
P_{i}$ from $P$.

A \textit{(voting) rule} is a function $f:\mathcal{P}%
^{n}\longrightarrow X$ selecting an alternative for each preference profile in $%
\mathcal{P}^{n}$. We assume throughout that a voting rule is an \textit{onto} function, i.e., for each $x \in X$ there is $P \in \mathcal{P}^n$ such that $f(P)=x.$  A rule $f:\mathcal{P}^{n}\longrightarrow X$ is 
\textit{tops-only }if for all $P,P^{\prime }\in \mathcal{P}^{n}$
such that $t (P_{i})=t (P_{i}^{\prime })$ for all $i\in N$, $%
f(P)=f(P^{\prime })$. In this paper, we will focus on tops-only rules.%


Given a rule $f:\mathcal{P}^{n}\longrightarrow X$ and $P_i\in \mathcal{P}$,  an alternative report $P_i' \in \mathcal{P}$ is a  \textit{(profitable) manipulation of rule $f$ at $P_i$} if  there is  a preference sub-profile  $P_{-i} \in \mathcal{P}^{n-1}$ such that 
\begin{equation*}
f(P_{i}^{\prime },P_{-i})P_{i}f(P_{i},P_{-i}).
\end{equation*}
A rule $f:\mathcal{P}^{n}\longrightarrow X$ is \textit{strategy-proof} on $%
\mathcal{P}^n$ if no agent has a manipulation. 

Other desirable  properties we  look at are the following. A  rule $f:\mathcal{P}%
^{n}\longrightarrow X$ is \textit{efficient }if for each $P\in \mathcal{P}^{n}$,
there is no $x\in X$ such that $xP_{i}f(P)$ for each $i\in N$. A rule $f:\mathcal{P}%
^{n}\longrightarrow X$ is \textit{anonymous }if it is invariant with respect to
the agents' names; namely, for each one-to-one mapping $\sigma :N\longrightarrow
N$ and each $P\in \mathcal{P}^{n}$, $f(P_{1},\ldots,P_{n})=f(P_{\sigma
(1)},\ldots,P_{\sigma (n)})$. A rule $f:\mathcal{P}%
^{n}\longrightarrow X$ is \textit{dictatorial }if there exists $i\in N$ such
that for each $P\in \mathcal{P}^{n}$, $f(P)=t(P_i)$.

The Gibbard-Satterthwaite Theorem states that a (onto) rule $f:%
\mathcal{P}^{n}\longrightarrow X,$ with $m\geq 3$, is
strategy-proof if and only if it is dictatorial \citep{gibbard1973manipulation, satterthwaite1975strategy}.  This negative result justifies the study of less demanding criteria of (lack of) manipulation when rules defined on the universal domain of preferences are considered. One such weakening of strategy-proofness is presented next.

\subsection{Obvious manipulations}

The notion of obvious manipulations has been introduced by \cite{troyan2020obvious} when applied to school choice models and later it has been  studied by \cite{aziz2021obvious} in the context of voting.
They try to describe those manipulations that are easily identifiable by the agents. To do this, it is important to specify how much information each agent has about other agents' preferences.  \cite{troyan2020obvious} assume that each agent has complete ignorance in this respect and, therefore, each agent focuses on the set of outcomes that can be chosen by the rule given its own report. 
 Now, a manipulation is obvious  if  the best possible outcome under the manipulation
is strictly better than the best possible outcome under truth-telling or the worst possible outcome under the manipulation is strictly better than the worst possible outcome under
truth-telling.

Before we present the formal 
 definition, we present some notation. Given a preference $P_{i} \in \mathcal{P}$, the \textit{option set} left open by $P_{i}$ at $f$  is 
\begin{equation*}
O^f(P_{i})=\{f(P_{i},P_{-i})\in X:P_{-i}\in \mathcal{P}^{n-1}\}.\footnote{\cite{barbera1990strategy} were the first to use option sets in the context of preference aggregation.}
\end{equation*}
 Given  $Y \subseteq X,$ denote by $B(P_i,Y)$ to the best alternative in $Y$ according to preference $P_i$, and by $W(P_i,Y)$ to the worst alternative in $Y$ according to preference $P_i$. 

\begin{definition}\label{def OM}{\citep{troyan2020obvious}}
Let $f:\mathcal{P}^n \longrightarrow X$ be a rule, let $P_i \in \mathcal{P},$ and let  $P_{i}^{\prime } \in \mathcal{P}$ be a  
 profitable manipulation  of  
$f$ at $P_i$. A manipulation $P_i'$ is \textbf{obvious} if
\begin{equation}\label{cond worst}
W(P_i,O^f(P_{i}^{\prime })) \ P_{i} \  W(P_i, O^f(P_{i}))
\end{equation}
or 
\begin{equation}\label{cond best}
B(P_i,O^f(P_{i}^{\prime })) \ P_{i} \ B(P_i, O^f(P_{i})).  
\end{equation}
The rule $f$ is \textbf{not obviously manipulable (NOM)} if it does
not admit any obvious manipulation.
\end{definition}

\section{Main theorem}\label{section main}

As previously mentioned, in this paper we focus in obvious manipulations of tops-only rules. Because of their simplicity,  tops-only voting rules  are both useful in practice and extensively studied in the literature. Now, assume we have a class of strategy-proof rules defined on a restricted domain of preferences and we want to study whether those rules are obviously manipulable in the universal domain of preferences. Then, it is natural to require tops-onlyness since, under mild assumptions, strategy-proofness implies tops-onlyness \citep[see, for example,][]{chatterji2011tops, weymark2008strategy}. Furthermore, obvious manipulations  allow us to discriminate among  different tops-only rules, while other recently studied weakenings of strategy-proofness
are incapable to do this because, under tops-onlyness, those concepts become equivalent to strategy-proofness \citep[see, for example,][]{arribillaga2022regret}.\footnote{\cite{arribillaga2022regret} assume that an agent knows a specific type of information about the preferences of the other individuals
conditional on his own declaration of preferences, and focus on manipulations for which the worst possible
outcome under the deviation consistent with the agent's information is strictly better than the outcome under
truth-telling. It can be seen that, under tops-onlyness, the existence of a manipulation of the aforementioned type is equivalent to the existence of a classical manipulation.} Obvious manipulations of non-tops-only rules have been thoroughly studied elsewhere \citep[see][]{aziz2021obvious}. In this section we provide a necessary and sufficient condition for a tops-only rule to be NOM.

In order to obtain our main result, we first need to define when an agent has veto power. An agent vetoes an alternative if there is a preference report of the agent that forces the rule to never select such alternative. Formally,

\begin{definition} Let $f:\mathcal{P}^n \longrightarrow X$ be a rule and let $i \in N,$  $x \in X,$ and $P_i \in \mathcal{P}$. Agent  \textbf{$\boldsymbol{i $ vetoes  $x$ via $P_{i}}$} if $x\notin O^f(P_{i}).$ 
\end{definition}
Denote by $V_{i}$ the set of all alternatives that agent $i$ vetoes via some preference. Given $x\in V_{i}$, let $%
\mathcal{V}_{i}^{x}=\{P_{i} \in \mathcal{P} :$ $i$ vetoes $x$ via $P_{i}\}$ be the set of all
preferences by which $x$ is vetoed by agent $i$. As it is observed by \cite{aziz2021obvious}, if a rule $f$ has no vetoers then it is NOM, because $O^f(P_i)=X$ for each $i \in N$ and each $P_i \in \mathcal{P}$. However, as the next example shows, there are rules that satisfy NOM and admit (many) vetoers.\footnote{For an example of a non-tops-only rule that satisfies NOM and admits vetoers, see Lemma 5 in \cite{aziz2021obvious}.} 

\begin{example}
Let $a \in X$ and consider the \emph{status quo rule at $a$}, $f^a:\mathcal{P}^n \longrightarrow X$, defined as $$f^a(P)=\begin{cases} x & \text{if } t(P_i)=x \text{ for each }i \in N \\
a & \text{otherwise.}\\\end{cases}$$ Observe that, if $P_i \in \mathcal{P}$ is such that $t(P_i)=a$,  then $O^{f^a}(P_i)=\{a\}$ and thus $V_i=X \setminus \{a\}.$ Furthermore, for any $P_i \in \mathcal{P},$ $O^{f^a}(P_i)=\{a, t(P_i)\}$. Therefore,  
$W(P_i, O^{f^a}(P_{i}))=a R_i W(P_i,O^{f^a}(P_{i}^{\prime }))$ 
and
$B(P_i, O^{f^a}(P_{i}))=t(P_i) R_i B(P_i,O^{f^a}(P_{i}^{\prime }))$ for each $P_i' \in \mathcal{P}.$ Hence, $f^a$ is NOM.          
\end{example}

The previous example shows that the condition of a rule having no vetoers is far from being necessary for the rule to be NOM. What turns out to be important is what kind of vetoes are admissible by the rule.  Next, we introduce a particular type  of veto power  that allows us to  state a necessary and sufficient condition for NOM in tops-only rules.      
We say that the veto of an alternative by an agent is strong if the report of \emph{any} preference with top different from the alternative forces the rule to never select it. Formally,

\begin{definition} Let $i \in N$ and  $x \in X$.  Agent \textbf{$\boldsymbol{i$ strongly vetoes $x}$} if $\mathcal{V}_i^x=\{P_i \in \mathcal{P}: t(P_i) \neq x\}.$ 
\end{definition}
Denote by $SV_i$ the set of all alternatives strongly vetoed by agent $i$. Note that $SV_i \subseteq V_i$ for each $i \in N$. 
 Clearly, the sets $V_{i}, SV_i$ and $\mathcal{V}^x_{i}$ depend on $f$ but we omit this reference to ease notation.

\begin{theorem}\label{Main Theorem}
A tops-only rule is NOM  if and only if
every veto is a strong veto, i.e., $SV_i=V_i$ for each agent $i\in N$.\footnote{When $m=2$ (two alternatives), the condition $SV_i=V_i$ is equivalent to the simple condition that agent $i$ does not veto $t(P_i)$ with $P_i$.}
\end{theorem}
\begin{proof}
Let $f:\mathcal{P}^n \longrightarrow X$ be a tops-only rule.

\noindent ($\Longrightarrow$) Assume there is $i \in N$ such that $V_i \neq SV_i.$ Since $SV_i \subseteq V_i$, $V_i \neq \emptyset$ and there is $x \in X$ such that $x \in V_i \setminus SV_i$. Then,  
\begin{equation}\label{top no es x}
\mathcal{V}_i^x\neq \{P_i \in \mathcal{P}: t(P_i)\neq x\}.    
\end{equation}
By \eqref{top no es x}, there are two cases to consider:
\begin{enumerate}
    \item[$\boldsymbol{1}.$] \textbf{There is $\boldsymbol{P'_i \in \mathcal{V}_i^x$ is such that $t(P'_i)=x}$}. By ontoness, there is $(P_i, P_{-i})\in\mathcal{P}^n$ such that $f(P_i, P_{-i})=x.$ Furthermore, $t(P_i)\neq x$, since otherwise $t(P_i)=x$ and tops-onlyness would imply $P_i' \notin \mathcal{V}_i^x$.  Let $\overline{P}_{i} \in \mathcal{P}$ be such that $t(\overline{P}_{i})=t
(P_{i})$ and $b(\overline{P}_{i})=x$.  By tops-onlyness, $f(\overline{P}_{i},P_{-i})=x.$ Since 
$P'_i \in \mathcal{V}_i^x$, $x\notin
O^f(P_{i}^{\prime }).$ Then, $f(P_{i}^{\prime },P_{-i})\neq x$ and therefore $%
f(P_{i}^{\prime },P_{-i})\overline{P}_{i}x=f(\overline{P}_{i},P_{-i})$, implying that $P_i'$ is a profitable manipulation of $f$ at $P_i$. Furthermore, as $x\notin O^f(P_{i}^{\prime })$, 
\begin{equation*}
W(\overline{P}_i, O^f(P_{i}^{\prime })) \ \overline{P}_{i} \ x= W(\overline{P}_i, O^f(\overline{P}_{i})).
\end{equation*}%
Thus, $P_{i}^{\prime }$ is an obvious manipulation of $f$. 
\item[$\boldsymbol{2}.$] \textbf{There is $\boldsymbol{P_{i} \in \mathcal{P}$ such
that $t (P_{i})\neq x$ and $P_{i}\notin \mathcal{V}_{i}^{x}}$}.
As $P_{i}\notin \mathcal{V}_{i}^{x}$, there
is $P_{-i} \in \mathcal{P}^{n-1}$ such that $f(P_i, P_{-i})=x.$ Let $\overline{P}_{i} \in \mathcal{P}$ be such that $t(\overline{P}_{i})=t
(P_{i})$ and $b(\overline{P}_{i})=x$.  By tops-onlyness, $f(\overline{P}_{i},P_{-i})=x.$ Since $%
x\in V_{i},$ there is $P_{i}^{\prime } \in \mathcal{P}$ such that 
 $P_i' \in \mathcal{V}_i^x,$ and the proof follows as in the previous case.

\end{enumerate}

\noindent ($\Longleftarrow$) Let $i \in N$ be such that $V_i=SV_i.$ We will prove that  agent $i$ has no obvious manipulations. If $V_{i}=\emptyset ,$ the proof is
trivial. Assume that $V_{i}\neq \emptyset.$ For each $P_{i} \in \mathcal{P}$, as $V_i=SV_i,$ 
$$O^f(P_i)=(X\setminus V_i) \cup \{t(P_i)\}.$$
Then, for each $P'_i\in \mathcal{P}$,
$$B(P_i, O^f(P_i))=t(P_i)\ R_i\ B(P_i, O^f(P_i'))$$
and
$$W(P_i, O^f(P_i))=W(P_i, (X\setminus V_i) \cup \{t(P_i)\}) \ R_i\ W(P_i, (X\setminus V_i) \cup \{t(P_i')\}) =W(P_i, O^f(P_i')).$$
Hence, $i$ does not have an obvious
manipulation.\end{proof}


 Corollary \ref{cor main eff} shows that under efficiency and tops-onlyness NOM implies a very limited veto power: at most one agent can veto some alternatives or only one alternative can be vetoed by some agents. 

\begin{corollary}\label{cor main eff}
An efficient and tops-only rule is NOM if and only
if some of the following statements hold:
\begin{enumerate}[(i)]
    \item  There is at most one $i\in N$ such that $V_{i}\neq \emptyset $ and, moreover, $SV_i=V_i$.
    \item There is $y\in X$ such that  $SV_i=V_i \subseteq \{y\}$, for each $i \in N$.
    
\end{enumerate}
\end{corollary}

\noindent \begin{proof}
Let $f:\mathcal{P}^n \longrightarrow X$ be an efficient and tops-only rule. 

\noindent ($\Longrightarrow$) Assume both conditions  (i) and (ii) do not hold. Then, there are distinct $i,j\in N$ and distinct $x,y\in X$ such
that $x\in V_{i}$ and $y\in V_{j}.$ Now let $P \in \mathcal{P}^n$ be such that  $P_{i}:y,x,\ldots$ and  $P_{k}:x,y,\ldots$ for each $k \in N \setminus \{i\}$. By efficiency, $f(P)\in \{x,y\}.$ Therefore, $P_{i}\notin \mathcal{V}_{i}^{x}$ or $P_{j}\notin \mathcal{V}_{i}^{y}.$ So, by Theorem \ref{Main Theorem}, $f$ is not  NOM.

\noindent ($\Longleftarrow$) By Theorem \ref{Main Theorem} it is clear that either condition is  sufficient for $f$ to be  NOM. 
\end{proof}

Corollary \ref{cor main anon and eff}
states that, under efficiency and anonymity, non-obvious manipulability is equivalent to having at most one alternative vetoed and that, if there is one such alternative, the veto is unanimous.

\begin{corollary}\label{cor main anon and eff}
An efficient, anonymous and tops-only rule is
NOM if and only if  either $V_{i}=\emptyset$ for each $i \in N$ or there is $y \in X$
such that $SV_i=V_i=\{y\}$ for each $i \in N$. 
\end{corollary}
\begin{proof}
It follows from Corollary \ref{cor main eff} and anonymity. 
\end{proof}
\bigskip


\section{Applications}\label{section applic}

In this section, we apply Theorem \ref{Main Theorem} to study two classes of tops-only  voting rules in two separate (but related) voting problems. Our results allow us to discriminate those rules in each class that are non-obviously manipulable in the universal domain of preferences.

In the first problem, presented in subsection \ref{subsection median}, alternatives are endowed with a linear order structure. When preferences are single peaked over that order, the family of (generalized) median voting schemes encompass all tops-only and strategy-proof (onto) rules. In the second problem, presented in subsection \ref{subsection committees}, alternatives consist of \emph{subsets} of objects chosen from a fixed finite set. When preferences are separable, the class of voting by committees encompass all strategy-proof (onto) rules.



\subsection{Median Voter Schemes}\label{subsection median}

In this subsection assume that $X$ is an ordered set. Without loss of generality, assume $X=\{a,a+1,a+2,\ldots, b\}\subseteq \mathbb{N}$ and $%
b=a+(m-1).$ A preference $P_{i}\in \mathcal{P}$ is \emph{single-peaked} on $X$
if for all $x,y\in X$ such that $x< y<t (P_{i})$ or $t
(P_{i})<y< x$, we have $t (P_{i})P_{i}yP_{i}x$.
We denote the domain of all single-peaked preferences on $X$ by $%
\mathcal{SP}$. Note that $\mathcal{SP} \subsetneq \mathcal{P}$. 

\cite{moulin1980strategy} characterizes the family of strategy-proof and tops-only (onto) rules on the domain
of single-peaked preferences. This family contains many non-dictatorial
rules. For example, when $n$ is odd, consider the rule $f:\mathcal{P}^{n}\longrightarrow X$  that selects, for each preference profile $%
P=(P_{1},\ldots,P_{n})\in \mathcal{P}^{n}$, the median among the top
alternatives of the $n$ agents; namely, $f(P)=med\{t (P_{1}),\ldots,t
(P_{n})\}$.\footnote{%
Given a set of real numbers $\{x_{1},\ldots,x_{K}\}$, where $K$ is odd, define
its \textit{median} as $med\{x_{1},\ldots,x_{K}\}=y$, where $y$ is such that $%
|\{1\leq k\leq K :  x_{k}\leq y\}|\geq \frac{K}{2}$ and $|\{1\leq k\leq
K :  x_{k}\geq y\}|\geq \frac{K}{2}$. Since $K$ is odd the median is unique
and belongs to the set $\{x_{1},\ldots ,x_{K}\}$.\smallskip} This rule is anonymous, efficient, tops-only, and strategy-proof on $\mathcal{%
SP}$. All other rules in the family given by \cite{moulin1980strategy} are extensions of $f$. 
Following \cite{moulin1980strategy}, and before presenting the general result, we first
introduce the anonymous subclass 
 and characterize those rules which  are NOM.
After that,  we present the general class of all strategy-proof and tops-only rules on $\mathcal{SP}^{n}$
and characterize those that are NOM when operating on  domain $%
\mathcal{P}^{n}$.

\subsubsection{Anonymity}
 A rule $f:%
\mathcal{P}^{n}\longrightarrow X$ is a \textit{median voter scheme} if there
is a vector  $\alpha=(\alpha_{1},\ldots,\alpha_{n-1})\in X^{n-1}$ of $n-1$ fixed ballots such that, for each%
\textit{\ }$P\in \mathcal{P}^{n}$\textit{,}%
\begin{equation*}
f(P)=med\{t (P_{1}),\ldots,t (P_{n}),\alpha_{1},\ldots,\alpha_{n-1}\}.  \label{MVS}
\end{equation*}%
Without loss of generality, throughout the paper we assume that $\alpha_{1}\leq \ldots\leq \alpha_{n-1}$. When we want to stress the dependence of the scheme on the vector $\alpha$ of fixed ballots, we write $f^\alpha.$ Also, to simplify notation, we use $O^{\alpha} (P_{i})$ instead of $O^{f^\alpha} (P_{i})$. The characterization of median voter schemes is as follows: 
\begin{fact}{\citep{moulin1980strategy}}
A (onto) rule $f:\mathcal{SP}^n \longrightarrow X$ is strategy-proof, tops-only, and anonymous if and only if it is a median voter scheme.\footnote{The definitions of strategy-proofness, tops-onlyness and anonymity on an arbitrary subdomain $\mathcal{U}^n\subsetneq \mathcal{P}^n$ are straightforward adaptations of the definitions given for the universal domain.} 
\end{fact}



Median voter schemes are strategy-proof on the domain $\mathcal{SP}^{n}$ of
single-peaked preferences. However, when they operate on the larger domain $%
\mathcal{P}^{n}$ they may become manipulable. Then, all median voter schemes
are equivalent from the classical manipulability point of view. Next, we give a simple test to identify which median voter schemes
are NOM.

\begin{theorem}\label{teorema median}
A median voter scheme $f^{\alpha}:\mathcal{P}%
^{n}\longrightarrow X$ is NOM if and
only if 
\begin{equation*}
    \alpha_{1}\in \{a,a+1\}  \text{  and  } \alpha_{n-1}\in \{b-1,b\}.
\end{equation*}
\end{theorem}

In order to prove Theorem \ref{teorema median}, the following  remark and  lemma are useful.  

\begin{remark}\label{rem median}
By definition of option set and $f^\alpha$,  for each $P_{i} \in \mathcal{P}$,  
$$O^\alpha(P_{i})=
\begin{cases}
    \{t(P_i),t(P_i)+1,\ldots, \alpha_{n-1}\} & \text{if } t(P_i)<\alpha_1\\
    \{\alpha_1,\alpha_1+1,\ldots,\alpha_{n-1}\} & \text{if } \alpha_{1}\leq t(P_i)\leq \alpha_{n-1}\\
    \{\alpha_1,\alpha_1+1,\ldots,t(P_i)\} & \text{if } \alpha_{n-1}<t(P_i)\\
\end{cases}$$
To see this when $t(P_i)<\alpha_1$, take $x \in  \{t(P_i),t(P_i)+1,\ldots, \alpha_{n-1}\}$ and let $P_{-i} \in \mathcal{P}^{n-1}$ be such that $t(P_j)=x$ for each $j \in N\setminus\{i\}$. Then, $f^\alpha(P_i, P_{-i})=x$ an so $x \in O^\alpha(P_i).$ Furthermore, if $x < t(P_i)$,  let $P_{-i} \in \mathcal{P}^{n-1}$ be such that $t(P_j)=x$ for each $j \in N\setminus\{i\}.$ Then, $f^\alpha(P_i, P_{-i})=t(P_i)$ and, therefore, $f^\alpha(P_i,P_{-i}') \neq x$ for each $P_{-i}' \in \mathcal{P}^{n-1}$. Thus, $x \notin O^\alpha(P_i)$. Symmetrically, it can be proven that $x \notin O^\alpha(P_i)$ when $x > \alpha_{n-1}$.  The cases $ t(P_i)\in \{\alpha_1, \ldots, \alpha_{n-1}\}$ and $t(P_i)>\alpha_{n-1}$ follow similar arguments and are omitted.
\end{remark}

\begin{lemma}\label{lema median voter}
Let $f^{\alpha}:\mathcal{P}%
^{n}\longrightarrow X$ be a median voter scheme and let $i \in N$. Then, $x\in V_{i}$ if and only if 
either  $x<\alpha_1$ or $x>\alpha_{n-1}$. 
\end{lemma}

\noindent \begin{proof}
($\Longrightarrow$) Assume that $x\in V_{i}.$ Then, there is $%
P_{i}\in \mathcal{P}$ such that $x\notin O^\alpha(P_{i}).$ Thus, by Remark \ref{rem median}, either  $x<\alpha_1$ or $x>\alpha_{n-1}$. 

\noindent ($\Longleftarrow$) First, assume $x<\alpha_{1}$ and let $P_{i}\in \mathcal{P}$ be such that $t (P_{i})=\alpha_{1}.$ Then, $x \notin
\{\alpha_{1},\alpha_{1}+1,\ldots,\alpha_{n-1}\}=O^\alpha(P_{i}).$ Thus, $x\in V_{i}.$ Now, let $x>\alpha_{n-1}$ and let $P_{i} \in \mathcal{P}$ be such that $t (P_{i})=\alpha_{n-1}.$ Then, 
$x\notin \{\alpha_1,\alpha_1+1,\ldots,\alpha_{n-1}\}=O^\alpha(P_{i}).$ Thus,  $x\in V_{i}.$
\end{proof}

\medskip 

\noindent \emph{Proof of Theorem \ref{teorema median}}. Let $f^{\alpha}:\mathcal{P}%
^{n}\longrightarrow X$ be a median voter scheme and let $i\in N.$ The proof relies in the following two facts: 
\begin{equation}\label{fact1}
 x< \alpha_1 \text{ and } x \in SV_i \text{ if and only if } x=a,     
\end{equation}
and
\begin{equation}\label{fact2}
\alpha_{n-1}<x \text{ and } x \in SV_i \text{ if and only if } x=b.   
\end{equation}
To see \eqref{fact1}, assume first that $x \in SV_i$ and $a<x < \alpha_1.$ Let $P_{i} \in \mathcal{P}$ be such that $t
(P_{i})=a.$ Then, $x \in O^\alpha(P_{i})=\{a,a+1,\ldots,\alpha_{n-1}\},$ contradicting that  $x \in SV_i$. Next, assume that $x=a$.
Let $P_i \in \mathcal{P}$ be such that $t (P_{i})\neq a$. Then, by Remark \ref{rem median},  
$x\notin O^\alpha(P_{i}).$ Hence, $i \in SV_i$. Thus, \eqref{fact1} holds. The proof of \eqref{fact2} is symmetric and therefore it is omitted. 
 To complete the proof of the theorem, assume $f^{\alpha}$ is NOM. By Theorem \ref{Main Theorem} and Lemma \ref{lema median voter}, $SV_i=\{x\in X:x<\alpha_{1} \text{ or } x>\alpha_{n-1}\}$. By \eqref{fact1} and \eqref{fact2}, $\{x\in X:x<\alpha_{1} \text{ or } x>\alpha_{n-1}\}\subseteq\{a,b\}$. Therefore,  $\alpha_{1}\in \{a,a+1\}$  and $\alpha_{n-1}\in \{b-1,b\}$. Now assume that $\alpha_{1}\in \{a,a+1\}$  and $\alpha_{n-1}\in \{b-1,b\}$. Then, $\{x\in X:x<\alpha_{1} \text{ or } x>\alpha_{n-1}\}\subseteq\{a,b\}$. By Lemma \ref{lema median voter}, $V_i\subseteq\{a,b\}$ and, hence, by \eqref{fact1} and \eqref{fact2}, $SV_i=V_i$. Therefore, by Theorem \ref{Main Theorem}, $f^\alpha$ is NOM. \hfill $\square$

\subsubsection{General Case}

Now we present the characterization of all strategy-proof, tops-only (onto) rules on the domain of single-peaked preferences
for all $n\geq 2$. A generalized median voter scheme can be identified with a set of fixed ballots $\{p_{S}\}_{S\in 2^{N}}$ on $X=\{a, a+1, \ldots, b\}$, one for each subset of agents $S$. Then, for each preference profile, the generalized median voter scheme identified with $\{p_{S}\}_{S\in 2^{N}}$ selects the alternative $x \in X$ that is the
smallest one with the following two properties: (i) there is a subset of agents $S$ whose top alternatives are smaller than or equal to $x$, and (ii) the fixed ballot $p_S$ associated to $S$
is also smaller than or equal to $x$. Formally, we say that a collection $p=\{p_{S}\}_{S\in 2^{N}}$ is a 
\textit{monotonic family of fixed ballots} if: (i) $p_{S}\in X$\
for all $S\in 2^{N}$\ with\ $p_{N}=a$ and $p_{\varnothing }=b$, and (ii) $%
T\subseteq Q$ implies $p_{Q}\leq p_{T}$. 
A rule $f:\mathcal{P}^n \longrightarrow X$ is a \emph{generalized median voter scheme} if there exits a monotonic family of fixed ballots  $p=\{p_{S}\}_{S\in 2^{N}}$ such that, for each $P \in \mathcal{P}^n,$ 
\begin{equation*}
f(P)=\min_{S\in 2^{N}}\max_{j\in S}\{t (P_{j}),p_{S}\}.
\end{equation*}
 When we want to stress the dependence of the scheme on the collection $p$  of fixed ballots, we write $f^p.$ Also, to simplify notation, we use $O^{p} (P_{i})$ instead of $O^{f^p} (P_{i})$.
The characterization of generalized median voter schemes is as follows: 
\begin{fact}{\citep{moulin1980strategy}}\label{prop general median}
A (onto) rule $f:\mathcal{SP}^n \longrightarrow  X$ is strategy-proof and tops-only if and only if it is a generalized median voter scheme.   
\end{fact}

The dictatorial rules are strategy-proof and tops-only, therefore they are generalized median voting schemes. It is easy to see that if the agents $i$ is the dictator, then  $p_{\{i\}}=a$ and $p_{N\backslash \{i\}}=b$. Trivially these rules are NOM. Next, we give a simple test to identify which non-dictatorial generalized median voter schemes
are NOM.

\begin{theorem}\label{teo general median}
A non-dictatorial generalized median voter
scheme  $f^{p}:\mathcal{P}%
^{n}\longrightarrow X$ is NOM if and only if, for each $i\in N$, 
\begin{equation} \label{gmv}
    p_{N\setminus\{i\}}\in \{a,a+1\}  \text{ and }  p_{\{i\}}\in \{b-1,b\}.
\end{equation}
\end{theorem}

In order to prove Theorem \ref{teo general median}, the following remark and  lemma are useful. 

\begin{remark}\label{rem generalized median voter}
By monotonicity of $p,$ $p_{N\setminus\{i\}}\leq p_{S}$
for each $S$ such that $i\notin S$ and $p_{T}\leq p_{\{i\}}$ for each $T$ such
that $i\in T.$ Assume that $p_{N\setminus\{i\}}\leq p_{\{i\}}.$ By definition of option set and $f^p$,  for each $P_{i} \in \mathcal{P}$,  
$$O^p(P_{i})=
\begin{cases}
    \{t(P_i),t(P_i)+1,\ldots, p_{\{i\}}\} & \text{if } t(P_i)<p_{N\setminus\{i\}}\\
    \{p_{N\setminus\{i\}},p_{N\setminus\{i\}}+1,\ldots,p_{\{i\}}\} & \text{if } p_{N\setminus\{i\}}\leq  t(P_i)\leq p_{\{i\}}\\
    \{p_{N\setminus\{i\}},p_{N\setminus\{i\}}+1,\ldots,t(P_i)\} & \text{if }p_{\{i\}} <t(P_i)\\
\end{cases}$$
The proof follows an  argument similar to the one used in Remark \ref{rem median}, with  $p_{N\setminus\{i\}}$ playing the role of $\alpha_{1}$ and $p_{\{i\}}$ playing
the role of $\alpha_{n-1}$. 
\end{remark}

\begin{lemma}\label{lem gen median}
Let $f^{p}:%
\mathcal{P}^{n}\longrightarrow X$ be a generalized median voter scheme and let $i \in N$. 
\begin{enumerate}[(i)]
    \item If $p_{N\setminus\{i\}}\leq p_{\{i\}}$, then $x\in V_{i}$ if and only if either $%
x<p_{N\setminus\{i\}}$ or $x>p_{\{i\}}.$
    \item If $p_{\{i\}}<p_{N \setminus \{i\}}$, then $V_{i}=X.$
\end{enumerate}
\end{lemma}
\begin{proof}
Let $f^{p}:%
\mathcal{P}^{n}\longrightarrow X$ be a generalized median voter scheme and let $i \in N$.

\noindent (i) The proof follows an  argument similar to the one used in the proof of Lemma \ref{lema median voter} (invoking Remark \ref{rem generalized median voter} instead of Remark \ref{rem median}). 

\noindent (ii) Let $x \in X$. There are two cases to consider:
\begin{itemize}
    \item[$\boldsymbol{1}.$] $\boldsymbol{p_{\{i\}}<x}.$  Let $P_{i}\in \mathcal{P}$ be such that $t
(P_{i})=p_{\{i\}}$. Then, $f^p(P_{i},P_{-i})\leq p_{\{i\}}$ for each $P_{-i}\in 
\mathcal{P}^{n-1}.$ This implies that $x \in V_{i}.$ 
    \item[$\boldsymbol{2}.$] $\boldsymbol{x\leq p_{\{i\}}<p_{N\setminus \{i\}}}.$ Let $P_{i}\in \mathcal{P}$
such that $t (P_{i})=p_{N\setminus\{i\}}$. Then, $f^p(P_{i},P_{-i})\geq p_{N\setminus\{i\}}$
for each $P_{-i}\in \mathcal{P}^{n-1}$ (because $p_{N\setminus\{i\}}\leq p_{S}$ for each 
$S$ such that $i\notin S$). This implies that $x \in V_{i}.$

\end{itemize}
In both cases $x\in V_{i}$. Therefore, $V_i=X.$
\end{proof}

\noindent \emph{Proof of Theorem \ref{teo general median}}.  Let $f^{p}$ be a non-dictatorial generalized median voter scheme. 

\noindent ($\Longleftarrow$) The proof that condition \eqref{gmv} implies that $f^{p}$ is NOM
follows a similar argument to that of the  proof of Theorem \ref{teorema median}, with $p_{N \setminus \{i\}}$ and $p_{i}$ playing the role of $\alpha_1$  and $\alpha_{n-1}$, respectively. Therefore it is omitted.

\noindent ($\Longrightarrow$)  Assume that $f^p$ is NOM. First, assume there is an agent $i^\star \in N$ such that $p_{\{i^\star\}}<p_{N\setminus\{i^\star\}}.$ Then, by Lemma \ref{lem gen median} , $V_{i^\star}=X.$ By Theorem \ref{Main Theorem}, $SV_{i^\star}=X.$ Thus, agent $i^\star$ is a dictator, contradicting that $%
f^p$ is non-dictatorial. Therefore $p_{N\setminus\{i\}}\leq p_{\{i\}}$ for each $i\in N.
$ Now, the proof follows a similar argument to the proof of Theorem \ref{teorema median} and, 
therefore, it is omitted. \hfill $\square$ 

\subsection{Voting by Committees}\label{subsection committees}

Now assume that agents have to choose a\textit{\ subset of objects} from a set $K
$ (with $\left\vert K\right\vert \geq 2$)$.$ Then, in this case, $X=2^{K}$ and elements of $X$ are subsets of $K$. A generic element of $K$ is
denoted by $k$ and a generic element of $X$ is
denoted by $S$ .  As an example, think of objects as candidates
to become new members of a society that have to be elected by the current members of the society. \cite{barbera1991voting} characterize, on the restricted domain of separable preferences, the family of all strategy-proof (onto) rules as the class of voting by committees. A preference $P_{i}$ of agent $i$ is separable if the division
between good objects (those $k \in K$ such that $\{k\}P_{i}\emptyset $) and bad objects (those $k \in K$ such that $\emptyset
P_{i}\{k\}$) guides the ordering of subsets in the sense that adding a good
object leads to a better set, while adding a bad object leads to a worse set. Formally,  agent $i$'s preference $P_{i}\in \mathcal{P}$ on $2^{K}$ is \emph{%
separable} if for all $S\in 2^{K}$ and $k \notin S$, 
\begin{equation*}
S\cup \{k \}P_{i}x\text{ if and only if }\{k \}P_{i}\emptyset .
\end{equation*}
Let $\mathcal{S}$ be the set of all separable preferences on $2^{K}$.
Observe that for any separable preference its top is the subset consisting of all good
objects. That is, for any separable preference $P_{i}\in \mathcal{S},$%
\begin{equation*}
t(P_{i})=\{k \in K : \{k \}P_{i}\emptyset \}.
\end{equation*}

We now define the class of rules known as voting by
committees. Let $N$ be a set of agents and $k \in K$ be an object. A
\emph{committee $\mathcal{W}_{k }$ for $k$} is a non-empty set of
non-empty coalitions (subsets) of $N$, that satisfies the following
monotonicity condition:%
\begin{equation*}
 M\in \mathcal{W}_{k }\text{ and }M\subseteq M^{\prime
} \text{ imply } M^{\prime }\in \mathcal{W}_{k}.
\end{equation*}
A rule $f:\mathcal{P}^{n}\longrightarrow 2^{K}$ is a \textit{%
voting by committees} if for each $k \in K$ there is a committee $%
\mathcal{W}_{k }$ such that, for each\textit{\ }$P\in \mathcal{P}^{n}$%
\textit{,}%
\begin{equation*}
k \in f(P)\text{ if and only if }\{i\in N : k \in t(P_{i})\}\in 
\mathcal{W}_{k }.
\end{equation*}

By definition, these rules are tops-only and the selected subset of objects at each preference profile is obtained by analyzing one object at a time. Given a committee  $\mathcal{W}=\{W_k\}_{k \in K}$, let $f^{\mathcal{W}}$ be its associated voting by committees. Furthermore, and to ease notation, we write $O^\mathcal{W}(P_{i})$ instead of $O^{f^{\mathcal{W}}}(P_{i})$. 
\cite{barbera1991voting} characterize this class when it operates on the separable domain as
follows.

\begin{fact}{\citep{barbera1991voting}}\label{prop voting committees}
A (onto) rule $f:\mathcal{S}^{n}\longrightarrow 2^{K}$ is strategy-proof if and only if it is voting by
committees.
\end{fact}
Observe that, since voting by committees are  tops-only,  tops-onlyness is implied by strategy-proofness on the domain of separable preferences.

It is clear that dictatorial voting by committees are NOM. In these rules, the dictator  $i\in N$ is such that $i\in M$ for all $M\in \mathcal{W}_{k }$  and $\{i\}\in \mathcal{W}_{k }$ for all $k \in K$. Next, we give a simple test to identify which non-dictatorial voting by committees are NOM.

\begin{theorem}\label{teo committees}
A non-dictatorial voting by committees $f^{%
\mathcal{W}}:\mathcal{P}%
^{n}\longrightarrow 2^K$ is NOM if and only if for each  $k \in K$:  
\begin{enumerate}[(i)]
    \item  $\bigcap_{M \in \mathcal{W}_k} M=\emptyset$, and  
    \item $|M| \geq 2$  for each $M \in \mathcal{W}_k$. 
\end{enumerate}
\end{theorem}

In order to prove Theorem \ref{teo committees}, the following remark and  lemma are useful.

\begin{remark}\label{rem committees}
Let $i \in N$. If  $S \in 2^K$ is such that
\begin{enumerate}[(i)]
    \item for each $k \in S$, $i \notin \bigcap_{M \in \mathcal{W}_k
}M$ and
    \item for each $k \notin S$, $\{i\}\notin 
\mathcal{W}_{k}$,
\end{enumerate}
then 
\begin{equation}\label{eq rem}
S\in O^\mathcal{W}(P_{i}) \text{ \ for each \ }P_i \in \mathcal{P}.    
\end{equation}
To see this, let $P_i \in \mathcal{P}$ and let $P_{-i}\in \mathcal{P}^{n-1}$ be such that $t(P_{j})=S$ for each $j\in N\setminus \{i\}$. Observe that condition (i) for $S$ implies  that  $N \setminus \{i\} \in \mathcal{W}_k$  for each $k \in S$ and, thus,  $S \subseteq f^\mathcal{W}(P_i, P_{-i}).$ Moreover, $k \notin S$ implies, by condition (ii) for $S$, that  $k \notin f^\mathcal{W}(P_i, P_{-i})$. Therefore, $f^{\mathcal{W}}(P_{i},P_{-i})=S$  and \eqref{eq rem} holds.  
\end{remark}

\begin{lemma}\label{lem committees}
Let $f^{\mathcal{W}}:\mathcal{P}^{n}\longrightarrow 2^{K}$ be a voting by committees, and let $i \in N$. Then, $S \in V_{i}$ if and only if there is $k ^{\star }\in K$ such that either:
\begin{enumerate}[(i)]
    \item $%
k ^{\star }\in S$ and $i\in \cap_{M \in \mathcal{W}_{k^\star}} M$, or

\item $k ^{\star}\notin S$ and $\{i\}\in \mathcal{W}_{k ^{\star }}$%
.\footnote{%
In the context of voting by committees, when $i\in M$ for each $M\in \mathcal{W}%
_{k ^{\star }}$, it is usual to say that agent $i$ is a \textit{vetoer} of $%
k ^{\star }$ in the sense that alternative $k^\star$ must be in the top of agents $i$'s preference to be included in the outcome of the rule. Such  veto notion is different from the veto condition in the present paper. Our notion is used to describe when an agent vetoes a subset $S \subseteq 2^K$ in the sense that there is a preference $P_i$ for agent $i$ such that $S$ is never chosen when $i$ declares $P_i$.
}
\end{enumerate}

\end{lemma}
\begin{proof} Let $f^{\mathcal{W}}:\mathcal{P}^{n}\longrightarrow 2^{K}$ be a voting by committees and let $i \in N$.

\noindent ($\Longrightarrow$) 
Assume that $S \in V_{i}.$ Then, there is $%
P_{i}\in \mathcal{P}$ such that $S \notin O^\mathcal{W}(P_{i}).$ Thus, by Remark \ref{rem committees}, there is $k ^{\star}\in K$ such that either $k ^{\star }\in S$ and $i\in \cap_{M \in \mathcal{W}_{k^\star}} M$, or $k ^{\star }\notin S$
and $\{i\}\in \mathcal{W}_{k ^{\ast }}$.

\noindent ($\Longleftarrow$) There are two cases to consider:
\begin{enumerate}
\item[$\boldsymbol{1}.$]\textbf{ There is $\boldsymbol{k ^{\star }\in K$  such that $k^{\star }\in S$ and $i\in \cap_{M \in \mathcal{W}_{k^\star}} M}$}. Let $P_{i}\in \mathcal{P}$ be such that $k ^{\star} \notin t (P_{i}).$   Then, $k ^{\star }\notin f(P_{i},P_{-i})$ and, therefore, $%
f(P_{i},P_{-i})\neq S$ for each $P_{-i}\in \mathcal{P}^{n-1}.$ Thus, agent $i$
vetoes $S$ with $P_{i}.$ Hence, $S \in V_{i}.$ 

\item[$\boldsymbol{2}.$]\textbf{There is $\boldsymbol{k^{\star}\in K$ such that $k
^{\star}\notin S$ and $\{i\}\in \mathcal{W}_{k ^{\star}}}$}. Let $%
P_{i}\in \mathcal{P}$ be  such that $k ^{\star} \in t (P_{i}).$ Then, $%
k ^{\star}\in f(P_{i},P_{-i})$ and, therefore,  $f(P_{i},P_{-i})\neq S$
for each $P_{-i}\in \mathcal{P}^{n-1}.$ Thus, agent $i$ vetoes $S$ with $P_{i}.$ Hence, $S \in V_{i}.$
\end{enumerate}
\end{proof}

\medskip 

\noindent \emph{Proof of Theorem \ref{teo committees}}. Let $f^\mathcal{W}:\mathcal{P}^n \longrightarrow 2^K$ be a non-dictatorial voting by committees. 

\noindent ($\Longleftarrow$) By Lemma \ref{lem committees}, (i) and  (ii) in Theorem \ref{teo committees} imply that $%
V_{i}=\emptyset$ for each $i\in N.$ Then, $SV_i=V_i$ for each $i \in N$ and, by Theorem \ref{Main Theorem}, $f^{%
\mathcal{W}}$ is NOM.

\noindent ($\Longrightarrow$) Let $f^{%
\mathcal{W}}$ be  NOM. Then, by Theorem \ref{Main Theorem}, $SV_i=V_i$ for each $i \in N$. The next claim states that $V_{i}=\emptyset$ for each $i\in N.$

\noindent \textbf{Claim: $\boldsymbol{V_{i}= \emptyset}$ for each $\boldsymbol{i \in N}.$} Assume, by contradiction, that there are $i \in N$ and $S \in X$ such that $S \in V_i$. We will show that such $i$ is a dictator, i.e., for each $k \in K$, 
\begin{equation}\label{dictator1}
\{i\}\in \mathcal{W}_{k },
\end{equation} and 
\begin{equation}\label{dictator2}
i \in \bigcap_{M \in \mathcal{W}_k} M. 
\end{equation}
By Lemma \ref{lem committees}, there are  two cases to consider:

\begin{enumerate}
\item[$\boldsymbol{1}$.]  \textbf{There is $\boldsymbol{k ^{\star }\in S$ and $i\in \bigcap_{M \in \mathcal{W}_{k^\star}} M}$}.  By Lemma \ref{lem committees}, $\{k ^{\star
}\}\in V_{i}.$ Assume \eqref{dictator1} does not hold. Then, there is $k \in K$ such that $\{i\} \notin \mathcal{W}_k.$ Let $P_{i} \in \mathcal{P}$ be such that $t (P_{i})=\{k ,k
^{\star}\}$. Then, $\{k ^{\star}\}\in O^\mathcal{W}(P_{i})$ and, therefore, $\{k ^{\star}\} \notin SV_i$. This contradicts $SV_i=V_i,$ so \eqref{dictator1} holds. Now, assume that \eqref{dictator2} does not hold. Then, there is $k \in K$ and $M\in \mathcal{W}_{k }$
such that $i\notin M.$ By Lemma \ref{lem committees},  $\{k ^{\star },k \}\in
V_{i}.$ Let $P_{i} \in \mathcal{P}$ be such that $t (P_{i})=\{k ^{\star }\}$. Then, $%
\{k,k^\star\}\in O^{\mathcal{W}}(P_{i})$. Thus,  $\{k, k^\star\} \notin SV_{i}.$ This contradicts $SV_i=V_i,$ so \eqref{dictator2} holds. 
Since both \eqref{dictator1} and \eqref{dictator2} hold, $i$ is a dictator.

\item[$\boldsymbol{2}$.] \textbf{There is $\boldsymbol{k ^{\star }\notin S$ and $\{i\}\in 
\mathcal{W}_{k ^{\star }}}$}.  By Lemma \ref{lem committees}, $\emptyset \in V_{i}.$  Assume \eqref{dictator1} does not hold. Then, there is $k \in K$ such that $\{i\} \notin \mathcal{W}_k.$  
 Let $P_{i} \in \mathcal{P}$ be such that $t (P_{i})=\{k\}$. Then, $\emptyset \in O^{\mathcal{W}}(P_{i})$ and, therefore, $\emptyset \notin SV_i$. This contradicts $SV_i=V_i,$ so \eqref{dictator1} holds.  Now, assume that \eqref{dictator2} does not hold. Then, there is $k \in K$ and $M\in \mathcal{W}_{k }$
such that $i\notin M.$ By Lemma \ref{lem committees},  $\{k \}\in
V_{i}.$ Let $P_{i} \in \mathcal{P}$ be such that $t (P_{i})=\emptyset$. Then, $%
\{k \}\in O^{\mathcal{W}}(P_{i})$. Thus,  $\{k \} \notin SV_{i}.$ This contradicts $SV_i=V_i,$ so \eqref{dictator2} holds. 
Since both \eqref{dictator1} and \eqref{dictator2} hold, $i$ is a dictator. 

\end{enumerate}
The fact that $i$ is a dictator contradicts  that $f^\mathcal{W}$ is non-dictatorial. Therefore, $V_i=\emptyset$ for each $i \in N.$ This finishes the proof of the Claim.

In order to complete the proof of the theorem, observe that the Claim and  Lemma \ref{lem committees} imply (i) and (ii) in Theorem \ref{teo committees}. \hfill $\square$

\bigskip

If we add anonymity to Fact \ref{prop voting committees} the class of voting by committees  must be
reduced to a relevant subclass of rules  which are called voting by quota. A voting by
committees is a \textit{voting by quota }if, for each $k \in K$, there is $q_{k}$, with $1 \leq q_k \leq n$, such that the associated committee $\mathcal{W}_k$ satisfies that
\begin{equation*}
M\in \mathcal{W}_{k }\text{ if and only if }\left\vert M\right\vert
\geq q_{k}.
\end{equation*}
Given  $q=\{q_k\}_{k \in K}$, let $f^{q}$ be its associated voting by quota. Next, we give a simple test to identify which voting by quota are NOM.

\begin{corollary}\label{cor quota}
A voting by quota $f^{q}:\mathcal{P}^{n}\longrightarrow 2^{K}$ is NOM if and only if $2 \leq q_k \leq n-1$ for each $k \in K$. 
\end{corollary}
\begin{proof}
It follows easily  by specifying Conditions (i) and (ii) in Theorem \ref{teo committees} to voting by quota.
\end{proof}

Corollary \ref{cor quota} is rather surprising. In general, rules in which quotas are either   $1$ or $n$ are the most robust to manipulation within all voting by quota from several standpoints \citep[for example, see][]{arribillaga2017comparing,fioravanti2022false}. Our result, in contrast, includes them within obviously manipulable ones.

\section{Final Remarks}\label{section final}

Table \ref{tabla caracterizaciones} summarizes our main findings about tops-only, median voter (MV), generalized median voter (GMV), voting by committees (VbC) and voting by quota (VbQ) rules.

\begin{table}[h] 
\small
\centering 
\begin{threeparttable}
\begin{tabular}{|c |l |c|}
\hline
Tops-only &  \hspace{7 pt}
$f$ NOM  \ \ $\Longleftrightarrow$ \ \ $\forall i \in N:$ \   $SV_i=V_i$   & Th. \ref{Main Theorem} \\
 \hline \hline
MV & \hspace{2 pt} 
$f^\alpha$ NOM \ \ $\Longleftrightarrow$ \ \ $\alpha_1 \in \{a, a+1\}$ and $\alpha_{n-1}\in \{b-1,b\}$ & Th. \ref{teorema median}  \\
\hline \hline
GMV$^\dag$ &  \hspace{2 pt} $f^p$ NOM \ \ $\Longleftrightarrow$ \ \ $\forall i \in N:$ \ $p_{N \setminus\{i\}} \in \{a, a+1\}$ and $p_{\{i\}}\in \{b-1,b\}$ & Th.  \ref{teo general median}\\
\hline \hline
VbC$^\dag$ &  $f^\mathcal{W}$ NOM \ \ $\Longleftrightarrow$ \ \ $\forall k \in K:$ \   $\bigcap_{M \in \mathcal{W}_k} M=\emptyset$  and  $|M| \geq 2$ $\forall M \in \mathcal{W}_k$ & Th. \ref{teo committees}\\
\hline \hline
VbQ & \hspace{3 pt} $f^q$ NOM \ \ $\Longleftrightarrow$ \ \ $\forall k \in K:$ $2 \leq q_k \leq n-1$& Cor. \ref{cor quota}\\
\hline
\end{tabular}
\begin{tablenotes}\footnotesize
\item[$\dag$]The characterization applies to non-dictatorial rules. 
\end{tablenotes}
\end{threeparttable}
\caption{\emph{Summary of  results.}} \label{tabla caracterizaciones}
\end{table}

Three final remarks are in order. First, following the proof of Theorem \ref{Main Theorem}, it can be seen that the condition $SV_{i}=V_{i}$ for each $i\in N$ implies NOM even when tops-onlyness is removed. Although clearly it is not a necessary condition of NOM \citep[see the proof of Lemma 5 in][]{aziz2021obvious}.

Second, a manipulation is called \textit{worst-case obvious} if Condition \eqref{cond worst} in Definition \ref{def OM} holds and \textit{best-case obvious} if Condition \eqref{cond best} in Definition \ref{def OM} holds. In the proof of Theorem \ref{Main Theorem}, we show that  $SV_i\neq V_i$ implies that the rule has a worst-case obvious manipulation. Therefore, by Theorem \ref{Main Theorem}, if a rule does not have a worst-case obvious manipulation then it does not have a best-case obvious manipulation either. Therefore, for top-only rules, a manipulation is obvious  if and only if  it is a worst-case obvious manipulation.

Finally, we present some observations for the case in which $X$ is infinite. If adequate assumptions over the set of preferences are done in order that \eqref{cond worst} and \eqref{cond best} in Definition \ref{def OM} are well-defined, Theorem \ref{Main Theorem} is also valid in such context. For example, let  $X=[a,b]\subseteq \mathbb{R}$ and let  $\mathcal{U}$ be the set of all continuous preferences on $[a,b]$ with a unique top (indifferences between non-top alternatives are admitted). If $f^\alpha:\mathcal{U}^n \longrightarrow X$ is a median voter scheme and $P_i \in \mathcal{U}$, then the option set is given by:
$$O^\alpha(P_{i})=
\begin{cases}
    [t(P_i),\alpha_{n-1}] & \text{if } t(P_i)<\alpha_1\\
    [\alpha_1,\alpha_{n-1}] & \text{if } t(P_i)\in[\alpha_1,\alpha_1]\\
    [\alpha_1,t(P_i)] & \text{if } t(P_i)>\alpha_{n-1}\\
\end{cases}
$$
and, therefore, $O^\alpha(P_{i})$ is a closed interval. Hence, \eqref{cond worst} and \eqref{cond best} in Definition \ref{def OM} are well-defined. For these rules,  if  $a < \alpha_1$,  for any $x \in X$ such that  $a<x<\alpha_1$  an  argument similar to the one used in the proof of Lemma \ref{lema median voter} shows that $x \in V_i$; and an argument similar to the one used in the proof of  Theorem \ref{teorema median} shows that $x \notin SV_i$. The same is true if $\alpha_{n-1}<b$ for any $x \in X$ such that $\alpha_{n-1}<x <b.$   Therefore, we have the following simple characterization of median voter schemes when $X=[a,b]\subseteq \mathbb{R}$ . 
\begin{theorem}
\;
\begin{enumerate}[(i)]
\item A median voter
scheme $f^{\alpha}:\mathcal{U}^{n}\longrightarrow \lbrack a,b]$   is NOM if and only if $\alpha_{1}=a$ and $\alpha_{n-1}=b$.

\item A non-dictatorial generalized median voter scheme $f^{p}:\mathcal{U}^{n}\longrightarrow \lbrack a,b]$ is NOM if and only if $p_{N\backslash \{i\}}=a$ and $p_{\{i\}}=b$ for each $i\in N.$
\end{enumerate}

\end{theorem}

\bibliographystyle{ecta}
\bibliography{biblio-obvious}

\end{document}